\documentclass[11pt,a4paper]{article}

\usepackage{amsmath,amssymb}
\usepackage{amsmath}
\usepackage{color}
\usepackage[colorlinks,linkcolor=blue,citecolor=red]{hyperref}
\usepackage{tikz}\usepackage{tikz-cd}
\usetikzlibrary{matrix,arrows}

\DeclareMathSymbol{\bbbr}{\mathalpha}{AMSb}{"52}
\DeclareMathSymbol{\bbbc}{\mathalpha}{AMSb}{"52}

\newcommand\com[1]{}
\newcommand\qed{\phantom{\underline{y}}\hfill\hfill$\square$\medskip}

\newtheorem{theorem}{Theorem}

\newtheorem{corollary}{Corollary}

\newtheorem{definition}{Definition}

\newtheorem{lemma}{Lemma}

\newtheorem{proposition}{Proposition}

\newenvironment{proof}[1]{\textbf{Proof.} }{\qed \vspace{5pt}}

\begin{document}

\title{On a class of 2D integrable lattice equations}

\author{E.V. Ferapontov$^{1,2}$, I.T. Habibullin$^{2,3}$, M.N. Kuznetsova$^2$, V.S. Novikov$^1$}
     \date{}
     \maketitle
     \vspace{-5mm}
\begin{center}
$^1$Department of Mathematical Sciences \\ Loughborough University \\
Loughborough, Leicestershire LE11 3TU \\ United Kingdom \\
  \bigskip
  $^2$Institute of Mathematics\\ Ufa Federal Research Centre\\ Russian Academy of Sciences \\
	 112 Chernyshevsky Street, Ufa 450008 \\ Russian Federation\\
  \bigskip
	$^3$Bashkir State University \\
	32 Validy Street, Ufa 450076 \\ Russian Federation\\
  \bigskip
e-mails: \\[1ex] 
\texttt{E.V.Ferapontov@lboro.ac.uk}\\
\texttt{habibullinismagil@gmail.com} \\
\texttt{mariya.n.kuznetsova@gmail.com} \\
\texttt{V.Novikov@lboro.ac.uk}
\\

\end{center}

\medskip
\begin{abstract}

We develop a new approach to the classification of integrable equations of the form
$$
u_{xy}=f(u, u_x, u_y, \triangle_z u \triangle_{\bar z}u, \triangle_{z\bar z}u),
$$
where $\triangle_{ z}$ and $\triangle_{\bar z}$ are the forward/backward discrete derivatives. The following 2-step classification procedure is proposed:

\noindent (1) First we require that the dispersionless limit of the equation is integrable, that is, its characteristic variety defines a conformal structure which is Einstein-Weyl on every solution.

\noindent (2) Secondly, to the candidate equations selected at the previous step we apply the test of Darboux integrability of reductions obtained by imposing suitable cut-off conditions.

\bigskip

\end{abstract}

\newpage


\section{Introduction }\label{sec:intro}

In this paper we develop a new approach to the classification of integrable lattice type equations in 3D by combining the geometric approach of \cite{FerKrug} with the  test of \cite{HP1, HP2} based on the requirement of Darboux integrability 
of suitably reduced equations. As an illustration we classify integrable equations of the form
\begin{equation}
u_{xy}=f(u, u_x, u_y, \triangle_z u \triangle_{\bar z}u, \triangle_{z\bar z}u).
\label{f}
\end{equation}
Familiar examples of type (\ref{f})  include the Toda equation
\begin{equation}\label{Toda}
u_{xy}=e^{\triangle_{z\bar z}u}
\end{equation} 
and the equation 
\begin{equation}
u_{xy}=u_xu_y\frac{\triangle_{z\bar z}u}{\triangle_z u \triangle_{\bar z}u}
\label{FSY}
\end{equation}
 discussed in \cite{SY, Fer1}. We use the notation $\triangle_{ z}=\frac{T_z-1}{\epsilon}, \ \triangle_{\bar z}=\frac{1-T_{\bar z}}{\epsilon}$ for the forward/backward discrete derivatives and $\triangle_{z \bar z}=\frac{T_z+T_{\bar z}-2}{\epsilon^2}$ for the symmetrised second-order discrete derivative; here $T_z, T_{\bar z}$ are the forward/backward $\epsilon$-shifts in the variable $z$. Note that dispersionless limits of the above equations (obtained as $\epsilon \to 0$) coincide with the Boyer-Finley equation $u_{xy}=e^{u_{z z}}$ and the equation 
$u_{xy}=\frac{u_xu_y}{u_z^2}u_{z z}$, respectively. Both limits belong to the class of dispersionless integrable PDEs. 
Further integrable examples of type (\ref{f}) obtained in  \cite{HP1} include the equations
\begin{equation}
u_{xy}=(u_x-u)(u_y-u)\frac{\triangle_{z\bar z}u}{\triangle_z u \triangle_{\bar z}u}+u_x+u_y-u
\label{H1}
\end{equation}
and
\begin{equation}
u_{xy}=(u_x-u^2-1)(u_y-u^2-1)\frac{\triangle_{z\bar z}u}{\triangle_z u \triangle_{\bar z}u}+2u(u_x+u_y-u^2-1).
\label{H2}
\end{equation}
We emphasize that although the lattice equations (\ref{FSY}), (\ref{H1}) and (\ref{H2}) are essentially different, their dispersionless limits are equivalent: setting $u=e^{v+x+y}$ and $u=\tan (v+x+y)$ in the dispersionless limits of (\ref{H1}) and (\ref{H2}), respectively, we obtain the dispersionless limit of equation (\ref{FSY}) (in variable $v$).

The above examples suggest the following 2-step classification procedure:

\noindent (1) First we classify  integrable equation of the form
\begin{equation}
u_{xy}=F(u, u_x, u_y, u_z, u_{zz}),
\label{f1}
\end{equation}
which can be viewed as dispersionless limits of equations (\ref{f}) when $\epsilon \to 0$. This can be done by requiring that the characteristic conformal structure $[g]$ of equation (\ref{f1}), namely
\begin{equation}\label{g}
[g]=4F_{u_{zz}}dxdy-dz^2,
\end{equation}
is Einstein-Weyl on every solution of (\ref{f1}) (see section \ref{sec:EW} for the necessary details). The classification results are summarised in section \ref{sec:2}. 

\noindent (2) Secondly, replacing  $u_z$ and $u_{zz}$ in the equations obtained at the previous step  by $\sqrt { \triangle_z u \triangle_{\bar z}u}$ and $\triangle_{z\bar z}u$, respectively,  we obtain equations of type (\ref{f}) which, at this stage, are our candidates for integrability.  To these candidate equations we apply the test of Darboux integrability of  reductions obtained by imposing suitable cut-off conditions as proposed in \cite{HP1, HP2}. The necessary details are provided in section \ref{sec:3}.

\subsection{Dispersionless integrability and Einstein-Weyl geometry}\label{sec:EW}

Recall that  Einstein-Weyl geometry is  a triple 
$(\mathbb{D}, [g], \omega)$ where $\mathbb{D}$ is a symmetric connection, $[g]$ is a conformal structure and $\omega$ is a covector such that  \cite{Cartan}:

\noindent (a) the connection $\mathbb {D}$ preserves the conformal class: $\mathbb{D}[g]=0$; 

\noindent (b) the trace-free part of the symmetrized Ricci tensor of $\mathbb {D}$ vanishes. 

In coordinates, this gives
\begin{equation}
\mathbb{D}_kg_{ij}=\omega_k g_{ij}, ~~~ R_{(ij)}=\Lambda g_{ij},
\label{EW}
\end{equation}
where $\omega=\omega_kdx_k$ is a covector, $R_{(ij)}$ is the symmetrized Ricci tensor of $\mathbb{D}$,  and $\Lambda$ is some function.  Note that it is sufficient to specify $[g]$ and $\omega$ only, then the first set of equations (\ref{EW}) uniquely defines   connection $\mathbb{D}$. We recall that in three dimensions (in what follows we label coordinates $x_1, x_2, x_3$ as $x, y, z$), Einstein-Weyl equations  (\ref{EW}) are  integrable via   twistor-theoretic construction \cite{Hitchin}. 

It was observed in \cite{Ward, Calderbank,  Dun4, Dun7} that Einstein-Weyl structures naturally arise on solutions of second-order dispersionless integrable PDEs. Furthermore, it was pointed out in \cite{FerKrug} that the corresponding conformal structures $[g]$ are defined by the characteristic varieties of these PDEs, and the covectors $\omega$ can also be efficiently calculated. Thus, the characteristic variety of equation (\ref{f1}) is
\begin{equation}\label{char}
p_xp_y-F_{u_{zz}}p_z^2=0,
\end{equation}
and the conformal structure $[g]$ defined by the inverse matrix (to the matrix of quadratic form (\ref{char})) coincides with (\ref{g}). 

\noindent {\bf Example 1.} The Boyer-Finley equation $u_{xy}=e^{u_{zz}}$ gives rise to the Einstein-Weyl structure \cite{Ward}
$$
[g]=4e^{u_{zz}}dxdy-dz^2, \qquad \omega=2u_{zzz} dz.
$$

\noindent {\bf Example 2.} The equation $u_{xy}=\frac{u_xu_y}{u_z^2}u_{z z}$ gives rise to the Einstein-Weyl structure
$$
[g]=4\frac{u_xu_y}{u_z^2}dxdy-dz^2, \qquad \omega=2\left(\ln \frac{u_xu_y}{u_z^2}\right)_z dz.
$$

For second-order dispersionless PDEs, the Einstein-Weyl property of the characteristic conformal structure can be seen as an efficient and universal integrability test. It was shown in   \cite{CalKrug} that this property implies the existence of a dispersionless Lax representation in parameter-dependent commuting vector fields.

\subsection{Summary of the main results}
\label{sec:form}

Our main result is a complete list of integrable equations of type (\ref{f}). Although the final list does not contain new integrable equations, the main purpose of the paper is to illustrate a novel approach to the classification problem: we emphasise that the general class of equations (\ref{f}) could not be tackled by any of the currently existing techniques.

\begin{theorem} \label{T1}
Modulo elementary equivalence transformations (rescalings and interchanges $x\leftrightarrow y$), any integrable equation of type (\ref{f}) is reducible to one of the following normal forms:
$$
\begin{array}{c}
u_{xy}=e^{\triangle_{z\bar z}u},\\
\ \\
u_{xy}=\frac{u_xu_y}{u}+u{\triangle_{z\bar z}u},\\
\ \\
u_{xy}=u_x{\triangle_{z\bar z}u},\\
\ \\
u_{xy}=u_xu_y\frac{\triangle_{z\bar z}u}{\triangle_z u \triangle_{\bar z}u},\\
\ \\
u_{xy}=(u_x-u)(u_y-u)\frac{\triangle_{z\bar z}u}{\triangle_z u \triangle_{\bar z}u}+u_x+u_y-u,\\
\ \\
u_{xy}=(u_x-u^2-1)(u_y-u^2-1)\frac{\triangle_{z\bar z}u}{\triangle_z u \triangle_{\bar z}u}+2u(u_x+u_y-u^2-1).
\end{array}
$$ 
\end{theorem}

The proof of Theorem \ref{T1} is summarised in Section 4.

\section{Integrable dispersionless equations $u_{xy}=F(u, u_x, u_y, u_z, u_{zz})$}\label{sec:2}

The  characteristic conformal structure has the form
$$
[g]=4F_{u_{zz}}dxdy-dz^2.
$$
Note that we always assume that the associated covector $\omega$ depends on finite-order jets of the variable $u$. Futhermore, the jet order of $\omega$ should be  by one higher than that of the conformal structure $[g]$; since $[g]$ depends on second-order derivatives of $u$, covector $\omega$ should involve no more than third-order derivatives. 
There are two cases to consider when we calculate the Einstein-Weyl conditions (\ref{EW}). 

\medskip

\noindent {\bf Genuinely nonlinear case $F_{u_{zz} u_{zz}}\ne 0$}. One can show that in this case the covector $\omega$ has the form 
$$
\omega=\left(\frac{2}{3}\frac{F_{u_z}}{F_{u_{zz}}} +\frac{10}{3}\frac{D_z(F_{u_{zz}})}{F_{u_{zz}}}
-\frac{4}{3}\frac{D_z(F_{u_{zz}u_{zz}})}{F_{u_{zz}u_{zz}}} \right)dz
$$
where $D_z$ denotes the total $z$-derivative (note that $dx$ and $dy$ components of $\omega$ vanish identically).
The Einstein-Weyl equations lead to a system of differential constraints for $F$ which result in the following two integrable dispersionless equations:
\begin{equation}\label{e1}
\beta'(u)u_{xy}+\beta''(u)u_xu_y=\gamma e^{\beta'(u)u_{zz}+\beta''(u)u_z^2}+\delta
\end{equation}
and
\begin{equation}\label{e2}
\beta'(u)u_{xy}+\beta''(u)u_xu_y=\gamma e^{\beta'(u)u_{zz}+\beta''(u)u_z^2+\delta \beta'(u)u_z+\frac{2}{9}\delta^2\beta(u)};
\end{equation}
here $\beta(u)$ is an arbitrary function and $\gamma, \delta$ are constants (without any loss of generality one can set $\gamma=1$). Note that although  $\beta(u)$ can be eliminated by a change of variables $\tilde u=\beta(u)$, this only works at the dispersionless level and is not necessarily valid for the corresponding lattice equations obtained by replacing  $u_z$ and $ u_{zz}$ with 
$\sqrt{\triangle_z u \triangle_{\bar z}u}$ and $\triangle_{z\bar z}u$. Thus, at this stage we will keep $\beta(u)$ arbitrary.

\medskip

\noindent {\bf Quasilinear case $F_{u_{zz} u_{zz}}= 0$}. Let us set
$$
u_{xy}=\varphi(u, u_x, u_y, u_z)u_{zz}+\psi(u, u_x, u_y, u_z).
$$
One can show that in this case the covector $\omega$ has the form 
$$
\omega=(2D_z(\ln \varphi)+\alpha)dz
$$
where $D_z$ denotes the total $z$-derivative and $\alpha$ is some constant (as above, both $dx$ and $dy$ components of $\omega$ vanish identically).  There are several subcases depending on which variables the coefficient $\varphi$ depends upon (we will assume $\varphi \ne const$). In what follows, $\beta=\beta(u)$ is an arbitrary function of $u$ and $\alpha, \gamma, \delta$ are arbitrary constants (the constant $\alpha$ plays a distinguished role, it is the same as in $\omega$); we will agree that if $\alpha$ does not appear on  the right-hand side of an equation below then it equals to zero in the expression for the corresponding covector $\omega$ as well.

\medskip

\noindent {\bf Subcase 1:} coefficient $\varphi$ depends on $u$ only, $\varphi_u\ne 0$. In this case the Einstein-Weyl conditions lead to  the following  integrable dispersionless equation:
\begin{equation} \label{e3}
u_{xy}=\beta u_{zz}+\frac{3}{2}\alpha \beta u_z+\frac{\alpha^2\beta^2}{2\beta'}+\left(\frac{\beta'}{\beta}-\frac{\beta''}{\beta'}\right)u_xu_y+\frac{\beta \beta''}{\beta'}u_z^2.
\end{equation}

\medskip

\noindent {\bf Subcase 2:} coefficient $\varphi$ depends on $u, u_z$ only, $\varphi_{u_z}\ne 0$. In this case we have  three integrable dispersionless equations:

\begin{equation} \label{e4}
u_{xy}=\gamma e^{\beta u_z} \left(u_{zz}+\frac{\beta'}{\beta}u_z^2\right)+\frac{\delta}{\beta}-\frac{\beta'}{\beta}u_xu_y,
\end{equation}

\begin{equation} \label{e5}
u_{xy}= e^{\alpha \beta+\beta' u_z} \left(u_{zz}+\alpha u_z+\frac{\alpha}{2\beta'}
+\frac{\beta''}{\beta'}u_z^2\right)-\frac{\beta''}{\beta'}u_xu_y,
\end{equation}

\begin{equation} \label{e6}
u_{xy}= e^{\frac{1}{2}\alpha \beta +\beta'u_z} \left(u_{zz}+\frac{1}{2}\alpha u_z+\frac{\alpha}{\beta'}
+\frac{\beta''}{\beta'}u_z^2\right)-\frac{\beta''}{\beta'}u_xu_y.
\end{equation}

\medskip

\noindent {\bf Subcase 3:} coefficient $\varphi$ depends on $u, u_z, u_y$ only, $\varphi_{u_y}\ne 0$. In this case we have  four integrable dispersionless equations:

\begin{equation} \label{e7}
u_{xy}=\beta'u_y u_{zz}+\left(\frac{1}{2}\alpha^2\beta+\frac{3}{2}\alpha \beta' u_z+\beta''u_z^2 \right) u_y-\frac{\beta''}{\beta'}u_xu_y,
\end{equation}

\begin{equation} \label{e8}
u_{xy}=(\gamma+\beta u_y)\left(u_{zz}+\frac{\delta}{\beta}+\frac{\beta'}{\beta}u_z^2 \right)-\frac{\beta'}{\beta}u_xu_y,
\end{equation}

\begin{equation} \label{e9}
u_{xy}=\gamma e^{\frac{1}{2}\alpha \beta+\beta' u_z}u_y (\alpha+2\beta' u_{zz}+\alpha{\beta'}u_z+2\beta''u_z^2)-\frac{\beta''}{\beta'}u_xu_y,
\end{equation}

\begin{equation} \label{e10}
u_{xy}=\delta e^{\beta u_z}\left(u_y+\frac{\gamma}{\beta}\right) (\beta u_{zz}+\beta'u_z^2)-\frac{\beta'}{\beta}u_xu_y.
\end{equation}

\medskip

\noindent {\bf Subcase 4:} coefficient $\varphi$ depends on all four arguments $u, u_z, u_y, u_x$, we can assume $\varphi_{u_x}\ne 0, \ \varphi_{u_y}\ne 0$. In this case we have the following equations:

\begin{equation} \label{e11}
u_{xy}=\frac{2u_{zz}+(4\beta'-\alpha)u_z+2\beta \beta'-\alpha \beta}{2(u_z+\beta)^2}u_xu_y,
\end{equation}

\begin{equation} \label{e12}
u_{xy}=\frac{u_xu_y+\beta u_x}{(u_z+\gamma \beta)^2}u_{zz}   +\frac{(4\gamma \beta'-\alpha)u_z+2\gamma^2\beta \beta'-\alpha \gamma \beta}{2(u_z+\gamma \beta)^2}u_xu_y-\frac{2\beta'u_z^2+\alpha \beta u_z+\alpha \gamma \beta^2}{2(u_z+\gamma \beta)^2}u_x,
\end{equation}

\begin{equation} \label{e13}
\begin{array}{c}
u_{xy}=\frac{(u_x+\beta)(u_y+\delta \beta)}{(u_z+\gamma \beta)^2}u_{zz}   
+\frac{(4\gamma \beta'-\alpha)u_z+2\gamma^2\beta \beta'-\alpha \gamma \beta}{2(u_z+\gamma \beta)^2}u_xu_y 
-\frac{2\beta'u_z^2+\alpha \beta u_z+\alpha \gamma \beta^2}{2(u_z+\gamma \beta)^2}(u_y+\delta u_x+\delta \beta),
\end{array}
\end{equation}
note that dispersionless limits of equations (\ref{FSY}), (\ref{H1}) and (\ref{H2}) can be obtained from equation (\ref{e13}) with the choice  of constants $\gamma=\alpha=0,\  \delta=1$, and the functions  $\beta(u)=0$, $\beta(u)=-u$ and $\beta(u)=-u^2-1$, respectively.

We also have the following three equations involving hyperbolic functions:

\begin{equation} \label{e14}
u_{xy}=\beta' \frac{\beta'u_{zz}+\frac{1}{2}\alpha \beta'u_z+ \beta''u_z^2}{\sinh^2(\gamma+\frac{1}{2} \alpha \beta+\beta'u_z)}u_xu_y-\frac{\beta''}{\beta'}u_xu_y,
\end{equation}

\begin{equation} \label{e15}
u_{xy}=\frac{\beta u_{zz}+ \beta'u_z^2}{\sinh^2(\delta+\beta u_z)}u_x(\gamma+\beta u_y)-\frac{\beta'}{\beta}u_xu_y,
\end{equation}

\begin{equation} \label{e16}
u_{xy}=\frac{(\mu+\beta u_x)(\nu+ \beta u_y)}{\sinh^2(\delta+\beta u_z)}u_{zz} 
+ \frac{\beta'}{\beta}\frac{\mu \nu+\beta(\mu u_y+\nu u_x+\beta u_xu_y)}{\sinh^2(\delta+\beta u_z)}u_z^2
-\frac{\beta'}{\beta}u_xu_y.
\end{equation}

The last example (\ref{e16}) is the most generic: it contains three constant parameters $\mu, \nu, \delta$ and an arbitrary function $\beta(u)$. This finishes the classification of integrable dispersionless equations of type (\ref{f1}).

\section{Integrable lattice equations}\label{sec:3}


Here we describe a classification procedure of integrable  lattices based on Darboux-integrable reductions and Lie-algebraic ideas developed in \cite{HP1, HP2}. It will be more convenient for our purposes to represent  equation (\ref{f}) in the equivalent lattice form,
\begin{equation}
u_{n,xy}=f(u_n, u_{n,x}, u_{n,y}, (u_{n+1}-u_n)(u_n-u_{n-1}), u_{n+1}-2u_n+u_{n-1}),
\label{ff1}
\end{equation}
obtained by formally setting $\epsilon=1$ in the expressions for discrete derivatives. We will not distinguish between representations (\ref{f}) and (\ref{ff1}) in what follows.
Our approach to lattice equations (\ref{ff1}) is based on the two pivotal moments: 

\begin{itemize}
\item A 3D equation  is integrable if it admits a large set of 2D reductions integrable in the sense of Darboux. 

\item Darboux integrability of 2D systems can be investigated by a method based on the theory of Lie-Rinehart characteristic algebras. 

\end{itemize}

We emphasize that all  known integrable 3D lattice equations are integrable in the sense of the following definition:

\medskip

\begin{definition} \label{Def1}
A lattice of the form
\begin{equation}  \label{f_n}
u_{n,xy}=g(u_{n+1}, u_{n}, u_{n-1}, u_{n,x}, u_{n,y})
\end{equation}
is said to be  integrable if there exist locally analytic functions $\varphi$ and $\psi$ of two variables  such that for any choice of integers $N_1$, $N_2$ the hyperbolic type system 
\begin{eqnarray}
&&u_{N_1,xy} = \varphi(u_{N_1+1},u_{N_1}), \nonumber \\
&&u_{n,xy}=g(u_{n+1}, u_{n},  u_{n-1}, u_{n,x}, u_{n,y}),\qquad N_1 < n < N_2, \label{finite_sys} \\
&&u_{N_2,xy}=\psi(u_{N_2},u_{N_2-1}), \nonumber 
\end{eqnarray} 
obtained from lattice (\ref{f_n}) by imposing cut-off conditions at $n=N_1$ and $n=N_2$,  is integrable in the sense of Darboux.
\end{definition}

Recall that Darboux integrability means that system (\ref{finite_sys}) possesses $N_2 - N_1 + 1$ nontrivial integrals in both characteristic directions. The function $\bar{u}=(u_{N_1},  \ldots, u_{N_2})$ and its derivatives $\bar{u}_x$, $\bar{u}_y$, $\bar{u}_{xx}$,  $\bar{u}_{yy}$, etc., are taken as dynamical variables. 
By definition, a function $I(\bar{u}, \bar{u}_x, \bar{u}_{xx}, \ldots)$ depending on a finite set of  dynamical variables is an $x$-integral of system (\ref{finite_sys}) if  $D_y I = 0$  where $D_y$ is the operator of  total derivative with respect to the variable $y$. That is to say $I$ is found from the system
\begin{equation*}
YI=0, \quad X_i I = 0
\end{equation*}
where 
\begin{equation*}\label{Ydef}
X_i = \frac{\partial}{\partial u_{i,y}},\qquad Y=\sum_{i=N_1}^{N_2} \left(u_{i,y} \frac{\partial}{\partial u_i} + g_i \frac{\partial}{\partial u_{i,x}} + D_x(g_{i})\frac{\partial}{\partial u_{i,xx}} + \cdots  \right)
\end{equation*}
and $g_i=g(u_{i+1},u_i,u_{i-1},u_{i,x},u_{i,y})$.

Let us consider the Lie algebra  $L_y$ generated by the operators $Y$, $X_i$ over the ring $K$ of locally analytic functions of the dynamical variables  $\bar u_{y},\bar u,\bar u_x,\bar u_{xx},\dots$. To the standard operation $[Z, W]=ZW-WZ$ we add the following conditions: for any $Z,W\in L_y$ and $a,b\in K$ we require  (i)
$[Z,aW]=Z(a)W+a[Z,W]$ and (ii)
 $(aZ)b=aZ(b)$.
These conditions mean that if $Z\in L_y$ and $a\in K$ then $aZ\in L_y$. The algebra $ L_y $ defined in this way is called the Lie-Rinehart algebra \cite{Rinehart}, \cite{Million}. We will also call it the characteristic algebra in  $y$-direction. In a similar way the characteristic algebra $L_x$ is defined.

The algebra $L_y$ is of  finite dimension if it admits a finite basis  of operators $Z_1,Z_2,\dots,Z_k\in L_y$ such that an arbitrary element $Z\in L_y$ can be represented as their linear combination: 
$Z=a_1Z_1+a_2Z_2+\dots +a_kZ_k$;
here the coefficients are functions $a_1,a_2,\dots,a_k\in K$.

Our approach is based on the following key statement \cite{ZMHSbook, ZhiberK}:

\begin{theorem} \label{T2}
System (\ref{finite_sys}) admits a complete set of  $y$-integrals ($x$-integrals) if and only if its characteristic algebra $L_y$ (respectively,  $L_x$) is of  finite dimension.
\end{theorem}

\begin{corollary}
System (\ref{finite_sys}) is integrable in the sense of Darboux if both characteristic algebras $L_x$ and $L_y$ are of  finite dimension.
\end{corollary}

Characteristic Lie algebras provide an effective method for classifying integrable cases of  lattice (\ref{f_n}). Here we recall the necessary results obtained by this method.

\begin{proposition}  \label{P1}
(see \cite{HP2}). Integrable equation of the form 
\begin{multline}  \label{Prop1_eq_gen}
u_{n,xy} = s(u_{n+1}, u_n, u_{n-1})u_{n,x}u_{n,y} + \beta(u_{n+1}, u_n, u_{n-1}) u_{n,x} + \\
+\gamma(u_{n+1}, u_n, u_{n-1}) u_{n,y} + \delta(u_{n+1}, u_n, u_{n-1}),
\end{multline}
with the coefficient $s$ satisfying the conditions $\frac{\partial s(u_{n+1}, u_n, u_{n-1})}{\partial u_{n \pm 1}} \neq 0$, can be reduced by a point transformation to one of the following forms:
\begin{eqnarray}
&& u_{n,xy}=\alpha_nu_{n,x}u_{n,y}, \quad \alpha_n = \frac{1}{u_n - u_{n-1}} - \frac{1}{u_{n+1}-u_n}=\frac{u_{n+1} - 2 u_n + u_{n-1}}{(u_{n+1}-u_n)(u_n - u_{n-1})}, \label{Prop1_eq1}\\
&& u_{n,xy} = \alpha_n(u_{n,x} - u_n)(u_{n,y}-u_n) +u_{n,x}+ u_{n,y} -u_n,  \label{Prop1_eq2}\\
&& u_{n,xy} = \alpha_n(u_{n,x} - u^2_n - 1)(u_{n,y} - u^2_n - 1) + 2 u_n(u_{n,x}+u_{n,y}-u^2_n - 1).  \label{Prop1_eq3}
\end{eqnarray}
\end{proposition}

\begin{proposition}  \label{P2}
(see \cite{Kuzn}). Integrable equation of the form
\begin{equation} 	\label{Prop2_eq_gen}
u_{n,xy} = g(u_{n+1}, u_n, u_{n-1}) u_{n,y} + \beta(u_{n+1}, u_n, u_{n-1}) u_{n,x} + \delta(u_{n+1}, u_n, u_{n-1}),  
\end{equation}
where the coefficient $g$ satisfies at least one of the conditions $\frac{\partial g}{\partial u_{n+1}}\neq 0$, $\frac{\partial g}{\partial u_{n-1}}\neq 0$, can be reduced by a point transformation to one of the following forms:
\begin{eqnarray}
&& u_{n,xy}=\left( e^{u_n - u_{n-1}} - e^{u_{n+1} - u_n} \right) u_{n,y}, 	\label{Prop2_eq1}\\
&& u_{n,xy} = \left( u_{n+1} - 2 u_n + u_{n-1} \right) u_{n,y}.						\label{Prop2_eq2}
\end{eqnarray}
\end{proposition}
Equations (\ref{Prop2_eq1}) and (\ref{Prop2_eq2}) were found earlier in \cite{SY}.

\begin{proposition}	\label{P3}
(see \cite{HabKuzS}). 
A lattice of the form 
\begin{equation}		\label{gg}
u_{n,xy} = g(u_{n+1}, u_n, u_{n-1}),
\end{equation}
which is integrable in the sense of Definition 1,  can be reduced by suitable rescalings to one of the following forms:
\begin{eqnarray}	
&& u_{n,xy} =  e^{\alpha u_n - \frac{\alpha}{2}m u_{n+1} - \frac{\alpha}{2} k u_{n-1}} + a(u_{n+1}, u_n) + b(u_n, u_{n-1}), \label{cond1}\\
&& u_{n,xy} =  e^{\alpha u_n} u_{n+1} u_{n-1} + a(u_{n+1}, u_n) + b(u_n, u_{n-1}), \label{cond2}\\
&& u_{n,xy} =  u_{n+1} u_{n-1} + a(u_{n+1}, u_n) + b(u_n, u_{n-1}), \label{cond3}\\
&& u_{n,xy} =  a(u_{n+1}, u_n) + b(u_n, u_{n-1}); \label{cond4}
\end{eqnarray}
here  $\alpha \neq 0$ and $m$, $k$ are positive integers.
\end{proposition}

The above propositions refer to rather special classes of lattice (\ref{f_n}). The general classification problem remains out of reach at present.

\section{Integrable lattice equations  $u_{xy}=f(u, u_x, u_y, \triangle_z u \triangle_{\bar z}u, \triangle_{z\bar z}u)$}

In this section we prove Theorem \ref{T1} by taking dispersionless equations (\ref{e1})-(\ref{e16}), replacing  $u_z$ and $ u_{zz}$ by 
$\sqrt{\triangle_z u \triangle_{\bar z}u}$ and $\triangle_{z\bar z}u$, respectively, and applying the  methods outlined in Section \ref{sec:3} to the resulting candidate equations. This is done in Sections \ref{sec:G}-\ref{sec:s4} below.

We start  by proving a useful statement concerning lattices of the form
\begin{equation}		\label{type1}
u_{n,xy} = s(u_n) u_{n,x} u_{n,y} + r(u_{n+1}, u_n, u_{n-1}).
\end{equation}

\begin{lemma}\label{lemma1} If equation (\ref{type1}) is integrable in the sense of Definition 1, then the  function
\begin{equation*}
\frac{\partial^2}{\partial u_{n+1} \partial u_{n-1}} r(u_{n+1}, u_n, u_{n-1})
\end{equation*}
is a quasi-polynomial in the variable $u_n$. 
\end{lemma}

\begin{proof}

We apply the method of characteristic algebras. The following 
characteristic operators  are associated with (\ref{type1}):
\begin{eqnarray*}
&&Y_i = \frac{\partial}{\partial u_i} + s(u_i) u_{i,x} \frac{\partial}{\partial u_{i,x}} + \cdots,  \\
&&R = \sum^{N_2}_{j = N_1}  r_j \frac{\partial}{\partial u_{j,x}} + \left( s(u_j) r_j u_{j,x} + r_{j,x} \right) \frac{\partial}{\partial u_{j,xx}} + \cdots ,
\end{eqnarray*}
where $r_j = r(u_{j+1}, u_j, u_{j-1})$. The commutators of the operator $D_x$  with the operators $Y_i$ and $R$ are given by the formulae
\begin{equation}  \label{DxYiDxR}
\left[ D_x, Y_i \right] = - s(u_i) u_{i,x} Y_i, \quad \left[D_x, R  \right] = - \sum^{N_2}_{j=N_1} r_j Y_j.
\end{equation}
Let us assume that $N_2 >> 1$, $-N_1 >> 1$ and concentrate on the subalgebra generated by the operators $Y_{-1}$, $Y_0$, $Y_1$ and $Z_{-1} = \left[Y_{-1}, R \right]$. By using (\ref{DxYiDxR}) and the Jacobi identity one can prove that 
\begin{equation*}
\left[D_x, Z_{-1}  \right] = -s(u_{-1}) u_{-1,x} Z_{-1} - Y_{-1}(r_{-2}) Y_{-2} + \left( R(s(u_{-1})u_{-1,x}) - Y_{-1} (r_{-1})\right) Y_{-1}  - Y_{-1}(r_{0}) Y_{0}.
\end{equation*}
Let us consider the operator $Z_{-1,1} = \left[Y_1, Z_{-1} \right]$ for which the following commutation formula holds:
\begin{equation*}
\left[ D_x, Z_{-1,1}  \right] = \left(s(u_1) u_{1,x} - s(u_{-1}) u_{-1,x} \right) Z_{-1,1} - r_{0, u_1 u_{-1}} Y_0.
\end{equation*}
Let us introduce the sequence
\begin{equation*}
T_0 = Z_{-1,1}, \quad  T_1 = \left[Y_0, T_0 \right], \quad T_2 = \left[ Y_0, T_1 \right],\ \dots,\  T_{k+1} = \left[Y_0, T_k \right],\ \dots
\end{equation*}
One has the following commutation formulae:
\begin{equation}	\label{DxT0}
\left[ D_x, T_0\right] = \left( s(u_1) u_{1,x} - s(u_{-1}) u_{-1,x} \right) T_0 - \frac{\partial r_{0, u_1 u_{-1}}}{\partial u_0} Y_0,
\end{equation}
\begin{equation}  \label{DxTk}
\left[ D_x, T_{k}  \right] = \left( s(u_1) u_{1,x} - s(u_{-1}) u_{-1,x} \right) T_k  - \frac{\partial^k r_{0, u_1 u_{-1}}}{\partial u^k_0} Y_0. 
\end{equation}
Due to the fact that the characteristic algebra must have finite dimension there should exist a natural $N$ such that the operator $T_{N+1}$  is linearly expressed through the operators $ T_N, \dots, T_0$, and the operators $T_N, \dots, T_0$ are linearly independent. Thus,
\begin{equation*} 
T_{N+1} + \lambda_N T_N + \cdots + \lambda_0 T_0 = 0
\end{equation*}
where  the coefficients $\lambda_i$, $i=0,1,\ldots N$, are functions of a finite set of  dynamical variables $\bar{u} = (u_{N_1}, u_{N_1-1}, \dots, u_{N_2})$, $\bar{u}_x, \bar{u}_{xx},$ etc. Commuting both sides of this equality with the operator $D_x$ and applying (\ref{DxT0}), (\ref{DxTk}) we obtain
\begin{multline*}
\left(s(u_1) u_{1,x} - s(u_{-1}) u_{-1,x} \right) \left(-\lambda_N T_N - \dots - \lambda_1 T_1 - \lambda_0 T_0  \right) - \frac{\partial^{N+1} r_{0, u_1 u_{-1}}}{\partial u^{N+1}_0} Y_0 + \\
+D_x(\lambda_n) T_N + \lambda_N \left( (s(u_1) u_{1,x} - s(u_{-1}) u_{-1,x}) T_N - \frac{\partial^N r_{0, u_1, u_{-1}}}{\partial u^N_0} Y_0 \right) + \\
\dots \\
+ D_x(\lambda_0) T_0 + \lambda_0 \left( (s(u_1) u_{1,x} - s(u_{-1}) u_{-1,x} ) T_0 - \frac{\partial r_{0, u_1 u_{-1}}}{\partial u_0} Y_0 \right) = 0.
\end{multline*}
Collecting the coefficients at the independent operators $T_N, T_{N-1}, \dots, T_0$ we obtain $D_x(\lambda_i)=0$, $i=0,...,N$. Therefore, $\lambda_i = const$. Comparing the coefficients at $ Y_0 $ we arrive at the ODE with constant coefficients for the function $r_{0, u_1 u_{- 1}}$:
\begin{equation*}		
\frac{\partial^{N+1} r_{0, u_1 u_{-1}}}{\partial u^{N+1}_0} + \lambda_N \frac{\partial^N r_{0, u_1, u_{-1}}}{\partial u^N_0} + \dots + \lambda_0 \frac{\partial r_{0, u_1 u_{-1}}}{\partial u_0} = 0.
\end{equation*}
Each solution $r_{0, u_1 u_{-1}}$ of this equation  is a quasi-polynomial in $u_0$. Lemma~\ref{lemma1} is proved. 
\end{proof}

Below, when studying  equations (\ref {e1})-(\ref {e16}), we assume that the  function $ \beta $ is analytic in a domain $D\subset\mathbb{C}$.

\subsection{Genuinely nonlinear case}
\label{sec:G}

Here we show that the only integrable equation of the form (\ref{f}) resulting from dispersionless equations (\ref{e1}), (\ref{e2}) is 
\begin{equation*}
u_{n,xy} = e^{\triangle_{z\bar z}u}.
\end{equation*}
Replacing  $u_z$ and $ u_{zz}$ by 
$\sqrt{\triangle_z u \triangle_{\bar z}u}$ and $\triangle_{z\bar z}u$, respectively,  we obtain equations of the following form:
\begin{equation}  \label{eqH1}
\beta'(u_n) u_{n,xy} + \beta''(u_n) u_{n,x} u_{n,y} = e^{\beta'(u_n)(u_{n+1} - 2 u_n + u_{n-1}) + \beta''(u_n)(u_{n+1}-u_n)(u_n - u_{n-1})} + \delta,
\end{equation}
\begin{multline}  \label{eqH2}
\beta'(u_n) u_{n,xy} + \beta''(u_n) u_{n,x} u_{n,y} = e^{\beta'(u_n)(u_{n+1} - 2 u_n + u_{n-1}) + \beta''(u_n)(u_{n+1}-u_n)(u_n - u_{n-1})} \times\\
\times e^{\delta \beta'(u_n) \sqrt{(u_{n+1}-u_n)(u_n - u_{n-1})} + \frac{2}{9}\delta^2 \beta(u_n) }.
\end{multline}
Note that these equations belong to the subclass (\ref{type1}). Thus we can apply Lemma~\ref{lemma1}.
 A simple analysis shows that the function $r_{n, u_{n+1} u_{n-1}}$ corresponding to the equation  (\ref{eqH2}) is a quasi-polynomial in $u_n$ only if $\delta = 0$. Consequently (\ref{eqH2}) is a particular case of (\ref{eqH1}) and we can focus on the latter. After the change of variables $v_n = \beta(u_n) \leftrightarrow u_n = \varphi(v_n) $ equation (\ref{eqH1}) takes the form
\begin{equation*}
v_{n,xy} = e^{\beta'(\varphi(v_n))( \varphi(v_{n+1}) - 2  \varphi(v_n) + \varphi(v_{n-1})) + \beta''(\varphi(v_n))(\varphi(v_{n+1}) - \varphi(v_n))(\varphi(v_n) - \varphi(v_{n-1}))} + \delta,
\end{equation*}
where  the function $\varphi$ is the inverse of $\beta$. Let us denote by $\bar{f}(v_{n+1}, v_n, v_{n-1})$ the right-hand side of this equation and evaluate the second-order derivative of $\bar{f}$:
\begin{multline*}
\bar{f}_{v_{n+1}, v_{n-1}} = \left(\bar{f}(v_{n+1}, v_n, v_{n-1})  - \delta \right)\varphi'(v_{n+1})\varphi'(v_{n-1})  \Bigl( -\beta''(u_n)+ \Bigr.\\
+\Bigl. \bigl( \beta'(u_n) + \beta''(u_n) (\varphi(v_n) - \varphi(v_{n-1}))  \bigr)\bigl( \beta'(u_n) - \beta''(u_n) (\varphi(v_{n+1}) - \varphi(v_{n}))  \bigr) \Bigr).
\end{multline*}
Due to Proposition~\ref{P3}, in the integrable case the function $\bar{f}(v_{n+1}, v_n, v_{n-1})$ is a quasipolynomial in all three variables such that
\begin{equation}	\label{barf}
\bar{f}_{v_{n+1}, v_{n-1}} = C e^{\alpha v_n - \frac{\alpha m}{2} v_{n+1} - \frac{\alpha k}{2} v_{n-1} }, 
\end{equation}
where $C$, $\alpha$ are arbitrary constants and $m$, $k$ are nonnegative integers related to each other in the following way: if $m = 0$ ($k=0$) then necessarily $k = 0$ ($m = 0$). Since equation (\ref{barf}) is complicated we  first impose a reduction  $v_{n+1} = v_n = v_{n-1}=v$ that obviously implies $\bar{f} - \delta = 1$. Now by comparing two representations for $\bar{f}_{v_{n+1}v_{n-1}}$ under the reduction we get (due to the equation $\beta'(u) = \frac{1}{\varphi'(v)}$)
\begin{equation*}
\beta'' = \beta'^2 (C e^{\alpha_1 \beta} - 1), 
\end{equation*}
where $\alpha_1 = \alpha \left( 1 - \frac{m}{2} - \frac{k}{2} \right)$. In terms of the function $\varphi(v)$ the equation takes the form
\begin{equation}  \label{phi_eq}
\frac{\varphi''}{\varphi'} = C e^{\alpha_1 v} - 1.
\end{equation}
For $\alpha_1 = 0$, $C \neq 1$ we get
\begin{equation*}
a) \, \varphi = e^{(C-1)v} \frac{c_2}{C - 1} + C_3.
\end{equation*}
For $\alpha_1 = 0$, $C = 1$ we find $\varphi'' = 0$ so that
\begin{equation*}
b) \, \varphi = c_1 v + c_2.
\end{equation*}
For $\alpha_1 \neq 0$, $C \neq 0$ the solution of (\ref{phi_eq}) is
\begin{equation*}
c) \, \varphi = \int{e^{-v} c_2 e^{\frac{C}{\alpha_1}e^{\alpha_1 v}} \mathrm{d} v}.
\end{equation*}
For $\alpha_1 \neq 0$, $C = 0$ we have
\begin{equation*}
d) \, \varphi = c_2 - e^{c_1 - v}.
\end{equation*}
According to Proposition~\ref{P3}, $\bar{f}$ is a quasi-polynomial in  the variables $v_{n+1}, v_n, v_{n-1}$. Let us take $v_n = v_{n-1}$ and note that the function 
\begin{equation*}
\bar{f}(v_{n+1}, v_n, v_{n}) = \mathrm{exp} \bigl( \beta'(u_n) \varphi(v_{n+1}) - \beta'(u_n) \varphi(v_n) \bigr)
\end{equation*}
has to be a quasi-polynomial in $v_{n+1}$. This is only possible in the case $b)$. Then equation (\ref{eqH1}) takes the form
\begin{equation*}
\frac{1}{c_1} u_{n,xy} = e^{\frac{1}{c_1} (u_{n+1} - 2 u_n + u_{n-1})} + \delta
\end{equation*}
which can be reduced to the first case of Theorem~\ref{T1} by the change of variables $ u_n \rightarrow c_1 u_n + c_1 \delta x y$.

\subsection{Subcases 1 and  2}
\label{sec:s12}

Below we prove that dispersionless equations (\ref{e3})--(\ref{e6}) give rise to the following integrable equation of the form (\ref{f}): 
\begin{eqnarray*}
u_{xy}=\frac{u_xu_y}{u}+u{\triangle_{z\bar z}u}.
\end{eqnarray*}
Replacing  $u_z$ and $ u_{zz}$ by 
$\sqrt{\triangle_z u \triangle_{\bar z}u}$ and $\triangle_{z\bar z}u$, respectively,  we obtain equations of the following form:
\begin{multline}	\label{H3}
u_{n,xy} = \left( \frac{\beta'(u_n)}{\beta(u_n)} - \frac{\beta''(u_n)}{\beta'(u_n)}  \right) u_{n,x} u_{n,y} +  \frac{\beta(u_n)\beta''(u_n)}{\beta'(u_n)} (u_{n+1} - u_n)(u_n - u_{n-1}) + \\ 
+ \beta(u_n) (u_{n+1}-2 u_n + u_{n-1}) 
+\frac{3}{2}\alpha \beta(u_n) \sqrt{(u_{n+1}-u_n)(u_n - u_{n-1})} + 
\alpha^2\frac{\beta^2(u_n)}{2 \beta'(u_n)},
\end{multline}
\begin{multline}	\label{H4}
u_{n,xy} = \gamma e^{\beta(u_n) \sqrt{(u_{n+1}-u_n)(u_n - u_{n-1})}} \times \\
\times \left( u_{n+1} - 2 u_n + u_{n-1} + \frac{\beta'(u_n)}{\beta(u_n)} (u_{n+1}-u_n)(u_n - u_{n-1}) \right) + \\
+ \frac{\delta}{\beta(u_n)} - \frac{\beta'(u_n)}{\beta(u_n)} u_{n,x} u_{n,y},
\end{multline}
\begin{multline}	\label{H5}
u_{n,xy} = e^{\alpha \beta(u_n) + \beta'(u_n) \sqrt{(u_{n+1}-u_n)(u_n- u_{n-1})}} \times\\
\times \left( u_{n+1} - 2 u_n + u_{n-1} + \frac{\alpha}{2 \beta'(u_n)} + \frac{\beta''(u_n)}{\beta'(u_n)} (u_{n+1} - u_n)(u_n - u_{n-1})  \right) - \\
- \frac{\beta''(u_n)}{\beta'(u_n)} u_{n,x} u_{n,y},
\end{multline}
\begin{multline}	\label{H6}
u_{n,xy} = e^{\frac{1}{2} \alpha \beta(u_n) + \beta'(u_n) \sqrt{(u_{n+1} - u_n)(u_n - u_{n-1})}} \times\\
\times \left( u_{n+1} - 2 u_n + u_{n-1} + \frac{1}{2} \alpha \sqrt{(u_{n+1}-u_n)(u_n - u_{n-1})} + \right. \\
\left. + \frac{\alpha}{\beta'(u_n)} + \frac{\beta''(u_n)}{\beta'(u_n)} (u_{n+1} - u_n) (u_n - u_{n-1}) \right) 
- \frac{\beta''(u_n)}{\beta'(u_n)} u_{n,x} u_{n,y}.
\end{multline}
Note that these equations belong to the subclass (\ref{type1}). Thus Lemma~\ref{lemma1} applies.
 A simple analysis shows that the function $r_{n, u_{n+1} u_{n-1}}$ corresponding to equations  (\ref{H5}), (\ref{H6}) is actually never a quasi-polynomial in $u_n$, since $\beta'$ does not vanish. Similarly we show that 
for the cases (\ref{H3}), (\ref{H4}) one should have $\alpha = 0$, $\gamma = 0$ (note that equation (\ref{H4}) with $\gamma=0$ does not contain the second-order discrete derivative, and therefore has a degenerate dispersionless limit). Therefore, equations (\ref{H5}), (\ref{H6}) are certainly non-integrable, while  integrable cases of the equations  (\ref{H3}), (\ref{H4}) reduce to
\begin{eqnarray}	
&&u_{n,xy} = \left( \frac{\beta'(u_n)}{\beta(u_n)} - \frac{\beta''(u_n)}{\beta'(u_n)}  \right) u_{n,x} u_{n,y} + \nonumber \\
 &&\hphantom{u_{n,xy}}+ \frac{\beta(u_n)\beta''(u_n)}{\beta'(u_n)} (u_{n+1} - u_n)(u_n - u_{n-1}) 
+ \beta(u_n) (u_{n+1}-2 u_n + u_{n-1}). \label{H3_1}
\end{eqnarray}
By a point transformation of the form $u_n = \varphi(v_n)$ we can eliminate in (\ref{H3_1}) the term containing the product $u_{n,x}u_{n,y}$. To this aim we choose $\varphi$ such that  $\beta(\varphi(v_n)) =  e^{v_n}$. Thus, 
\begin{equation*}
\beta(u_n)= e^{ v_n}, \quad \beta'(u_n) = \frac{1}{\varphi'(v_n)} e^{v_n}, \quad \beta''(u_n) = \frac{ e^{ v_n}}{\varphi'^2(v_n)} \left( 1 - \frac{\varphi''(v_n)}{\varphi'(v_n)} \right).		\label{beta_phi}
\end{equation*}
After that equation (\ref{H3_1}) takes the form
\begin{multline}  \label{H03}
v_{n,xy} = \frac{ e^{ v_n}}{\varphi'(v_n)} \left( \varphi(v_{n+1}) - 2 \varphi(v_n) + \varphi(v_{n-1}) \right) + \\
+ \frac{ e^{ v_n}}{\varphi'^2(v_n)} \left(1 - \frac{\varphi''(v_n)}{\varphi'(v_n)} \right) \left( \varphi(v_{n+1}) - \varphi(v_n)  \right) \left( \varphi(v_n) - \varphi(v_{n-1})  \right).
\end{multline}
We see that this equation belongs to the subclass (\ref{gg}). 
To study (\ref{H03}) we apply Proposition~\ref{P3}.
According to (\ref{cond1})--(\ref{cond4}) equation (\ref{H03}) can take one of the following forms:
\begin{eqnarray}
&&v_{n,xy} = e^{v_{n+1}} - 2 e^{v_n} + e^{v_{n-1}},		\label{H03_1}\\
&&v_{n,xy} = e^{v_n} (v^2_n + v_{n+1}v_{n-1} - v_{n+1}v_n - v_{n} v_{n-1} + 2 v_{n} - v_{n+1} - v_{n-1}),		\label{H03_2}\\
&&v_{n,xy} = -(m-2)e^{-v_n + \frac{m}{2} v_{n+1} + \frac{m}{2} v_{n-1}} + 2 (m-1) e^{(-1 + \frac{m}{2})v_n + \frac{m}{2} v_{n+1}} + \nonumber \\
&&\hphantom{v_{n,xy}} + 2 (m-1) e^{(-1 + \frac{m}{2})v_n + \frac{m}{2} v_{n-1}} + (2 - 3m) e^{(-1 + m)v_n}, 		\label{H03_3}
\end{eqnarray}
where $m$ is a positive integer. Equation (\ref{H03_1}) and equation (\ref{H03_3}) with $m = 1$  are known to be integrable. Equation (\ref{H03_3}) with $m = 1$ corresponds to the first case of Theorem~\ref{T1} (by the change of variables $ v_n \rightarrow 2 v_n$, $x \rightarrow 2x$).  Equation (\ref{H03_1}) corresponds to the second case of Theorem~\ref{T1}. The case $m = 2$ leads to a trivial degenerate equation, hence in our further study we suppose $(m-2)(m-1) \neq 0$.

Now let us prove that equation (\ref{H03_2}) is not integrable. To this aim we investigate the characteristic algebra of this equation generated by the operators
\begin{equation*}
X_j = \frac{\partial}{\partial v_j}, \quad Z = \sum^{N_2}_{i = N_1} \left(f_j \frac{\partial}{\partial v_{j,x}} + D_x (f_j) \frac{\partial}{\partial v_{j,xx}} + \cdots  \right),
\end{equation*}
which satisfy the relations
\begin{equation}  \label{DxXjDxZ}
\left[ D_x, X_j \right] = 0, \quad \left[D_x, Z  \right] = -\sum^{N_2}_{j=N_1} f_j X_j
\end{equation}
where $f_j = f(v_{j+1}, v_j, v_{j-1})$ is the right-hand side of equation (\ref{H03_2}) represented as $v_{j,xy} = f_j$.
We will need the following useful statement \cite{ShYam, ZMHSbook}:
\begin{lemma} \label{lemma2}
If a vector field of the form
\begin{equation*} 
Z = \sum_{i=0}^N z_{1,i} \frac{\partial}{\partial v_{i,x}} + z_{2,i} \frac{\partial}{\partial	v_{i,xx}} + \cdots
\end{equation*}
solves the equation $\left[ D_x, Z \right] = 0$, then $Z=0$.
\end{lemma}

Let us introduce the operators
\begin{equation*}
Z_0 = \left[X_0, Z\right], \quad W_0 = \left[ X_1, Z_0 \right], \quad W_1 = \left[X_0, W_0\right], \quad  W_2 = \left[X_0, W_1\right], \quad  W_3 = \left[X_0, W_2\right].
\end{equation*} 
Using the Jacobi identity and (\ref{DxXjDxZ}) one can find that
\begin{eqnarray*}
&&\left[D_x, Z_0  \right] = -X_0(f_{-1})X_{-1} - X_0(f_0)X_0 - X_0(f_1) X_1 = - e^{v_{-1}} (v_{-2} - v_{-1} - 1) X_{-1} - \nonumber\\
&&\hphantom{\left[D_x, Z_0  \right]}  - e^{v_0} (v^2_0 + v_1 v_{-1} - v_1 v_0 - v_0 v_{-1} + 4 v_0 - 2v_1 - 2 v_{-1} + 2) X_0 -\nonumber\\
&&\hphantom{\left[D_x, Z_0  \right]}  - e^{v_1} (v_2 - v_1 - 1) X_1,
\end{eqnarray*}
\begin{equation*}
\left[D_x, W_0  \right] = -e^{v_0} (v_{-1} - v_0 - 2)X_0 - e^{v_1} (v_2 - v_1 - 2) X_1,
\end{equation*}
\begin{equation*}
\left[D_x, W_k  \right] = -e^{v_0}(v_{-1} - v_0 - (k+2))X_0, \quad k = 1,2,3.
\end{equation*}
The next step is to introduce the operators $P = W_3 - W_2$, $Q = W_2 + 4 P$ for which the following formulae hold:
\begin{equation}  \label{DxPQ}
\left[D_x, P\right] = -e^{v_0} X_0, \quad \left[ D_x, Q \right] = -e^{v_0}(v_{-1} - v_0)X_0.
\end{equation}
The function $g = -e^{v_0}(v_{-1} - v_0)$ satisfying the relation $\left[ D_x, Q \right] = g X_0$ is annihilated by the operator $\Lambda(X_0) = (X_0 - 1)^2$. In other words, the polynomial $\Lambda(\lambda)$ corresponding to the operator $Q$ has a multiple root. Let us prove that in this case  the characteristic subalgebra generated by the operators $P$, $Q$ is infinite-dimensional. To this aim we construct a sequence of operators via the following formulae:
\begin{equation}  \label{seq_K}
P, \quad Q, \quad K_1 = \left[ P, Q \right], \quad K_2 = \left[P, K_1\right], \ldots, \quad K_{m+1} = \left[ P, K_m \right], \ldots
\end{equation}
In order to calculate the value of the expression $\left[ D_x, K_1 \right]$ we need to know $\left[X_0, P\right]$, $\left[X_0, Q  \right]$:
\begin{eqnarray*}
&&\left[D_x, \left[X_0, P  \right]  \right] = \left[X_0, \left[D_x, P\right]  \right] - \left[P, \left[D_x, X_0\right]  \right] = 
\left[X_0, -e^{v_0} X_0 \right] = - e^{v_0} X_0 = \left[D_x, P\right],\\
&&\left[ D_x, \left[X_0, Q  \right] \right] = \left[ X_0, \left[ D_x, Q \right] \right] - \left[ Q, \left[D_x, X_0  \right] \right] = \left[ X_0, \left[ D_x, Q \right] \right] = \nonumber\\
&&= \left[ X_0, e^{v_0} (v_0 - v_{-1}) X_0  \right] 
= e^{v_0} (v_0 - v_{-1} + 1) X_0 = \left[ D_x, Q\right] - \left[D_x, P  \right] = \left[ D_x, Q- P \right].
\end{eqnarray*}
Using Lemma~\ref{lemma2} we conclude that $\left[X_0, P  \right] = P$, $\left[ X_0,Q \right] =Q - P$. Now one can find
\begin{eqnarray}
&&\left[D_x, K_1  \right] = \left[D_x, \left[P, Q  \right]  \right] = \left[ P, \left[D_x, Q  \right] \right] - \left[Q, \left[ D_x, P \right] \right] = \nonumber\\
&&= \left[ P, e^{v_0} (v_0 - v_{-1})X_0 \right] - \left[ Q, -e^{v_0} X_0 \right] 
= -e^{v_0}(v_0 - v_{-1} - 1)P - e^{v_0}Q.  \label{DxK1}
\end{eqnarray}
Similarly we find 
\begin{eqnarray*}
\left[X_0, K_1  \right] = 2 K_1, \quad \left[D_x, K_2  \right] = -3 e^{v_0} K_1,\\
\left[ X_0, K_2 \right] = 3 K_2, \quad \left[ D_x, K_3 \right] = -6 e^{v_0} K_2.
\end{eqnarray*}
It can be proved by induction that the following formula holds:
\begin{equation}  \label{DxKm}
\left[ D_x, K_m \right] = -\frac{m(m+1)}{2} e^{v_0} K_{m-1}, \quad m \geq 2.
\end{equation}
Since the characteristic algebra generated by the operators $P$, $Q$ must be finite-dimensional, one of the two cases must hold: 

$1)$ there is an integer $M$ such that $K_{M+1}$ is expressed through the previous members of sequence (\ref{seq_K}):
\begin{equation}  \label{K_Mp1}
K_{M+1} = a_M K_M + a_{M-1} K_{M-1} + \cdots + a_1 K_1 + b_1 P + b_2 Q, \quad M \geq 1,
\end{equation}
where the operators $Q,P,K_1, \ldots K_{M}$ are linearly independent, the coefficients $a_i$, $b_k$ are functions of a finite set of dynamical variables $\bar{v}, \bar{v}_x, \bar{v}_{xx} \ldots $;

$2)$ the operator $K_1$ is linearly expressed through $P$, $Q$:
\begin{equation}  \label{K1}
K_1 = b_1 P + b_2 Q.
\end{equation}

Let us begin with case $1)$. We commute both sides of (\ref{K_Mp1}) with the operator $D_x$ and apply formula (\ref{DxKm}). Comparing the coefficients at the independent operators and collecting the coefficients at the operator $K_M$ we get the equality
\begin{equation*}
D_x(a_M) = -\frac{1}{2} (M+1)(M+2) e^{v_0}
\end{equation*}
which is never realized.

In case $2)$ by commuting both sides of (\ref{K1}) with the operator $D_x$ and by using formulae (\ref{DxK1}), (\ref{DxPQ}) we get a contradictory equation: 
\begin{eqnarray*}
-e^{v_0} (v_0 - v_{-1} - 1)P - e^{v_0} Q = D_x(b_1)P + D_x (b_2) Q - b_1 e^{v_0} X_0 + b_2 e^{v_0} (v_0 - v_{-1})X_0.
\end{eqnarray*}
Indeed, by comparing the coefficients at $Q$ we obtain the relation $D_x (b_2)=- e^{v_0}$ which has no solutions depending on a finite set of dynamical variables. Therefore this case  is also not realised. In other words, equation (\ref{H03_2}) is not integrable in the sense of Definition 1. 

Let us turn to equation (\ref{H03_3}) with $m\ne 1, 2$. It is easily verified that operator $Z_{0,-1} = \left[X_{-1}, \left[X_0, Z  \right]  \right]$  satisfies the following commutativity relation: 
\begin{eqnarray*}
&&\left[D_x, Z_{0,-1}  \right] = \left( A e^{-v_{-1} + \kappa(v_0 - v_{-2})} + B e^{\mu v_{-1} + \kappa v_{-2}}\right)X_{-1} +\nonumber\\
 &&+\left(A e^{-v_0 + \kappa(v_1 - v_{-1})} + B e^{\mu v_0 + \kappa v_{-1}}  \right)X_0 + 
 \left(A e^{-v_1 + \kappa(v_2 - v_0)} + B e^{\mu v_1 + \kappa v_0}  \right)X_1,
\end{eqnarray*}
where $\mu = -1 + \frac{m}{2}$, $\kappa = \frac{m}{2}$, $A = -(m-2)\frac{m}{2}$, $B = \frac{m(m-1)(m-2)}{2}$. Note that $A \neq 0$, $B \neq 0$.  We define the polynomial $\sigma(\lambda) = (\lambda + 1)(\lambda - \mu)(\lambda + \kappa)$ to obtain the formulae
\begin{eqnarray*}
\left[ D_x, \sigma(ad_{X_{-1}}) Z_{0, -1} \right] = B e^{\mu v_0} \sigma(X_{-1}) e^{\kappa v_{-1}} X_0, \\
\left[ D_x, \sigma(ad_{X_{1}}) Z_{0,-1} \right] = A e^{-\kappa v_{-1} - v_0} \sigma(X_1) e^{\kappa v_1} X_0.
\end{eqnarray*}
Evaluation gives $\sigma(X_{-1})e^{\kappa v_{-1}} = \frac{m	(m+2)}{2} e^{\kappa v_{-1}}$ and  $\sigma(X_{1})e^{\kappa v_{1}} = \frac{m	(m+2)}{2} e^{\kappa v_{1}}$. Let us introduce the notation $\sigma(ad_{X_{-1}})Z_{0, -1} = P_0$, $\sigma(ad_{X_1}) Z_{0,-1} = P_1$, $B_0 = B \frac{m(m+2)}{2}$, $A_0 = A \frac{m(m+2)}{2}$, and rewrite the last two equations in the form
\begin{equation*}
\left[ D_x, P_0 \right] = B_0 e^{\mu v_0 + \kappa v_{-1}}X_0, \quad \left[ D_x, P_1 \right] = A_0 e^{-v_0 + \kappa v_1 - \kappa v_{-1}} X_0.
\end{equation*}
According to the results of \cite{HabKuzS} the Lie-Rinehart algebra generated by the operators $P_0$, $P_1$ is finite-dimensional only if 
 $\mu = -1$, i.e., for $m = 0$. So we obtain the system of independent equations of the form $v_{n,xy} = 4 e^{-v_n}$ (each of these equations is the well-known integrable Liouville equation). This finishes the analysis of Subcases 1-2.  


\subsection{Subcase 3}
\label{sec:s3}

Among dispersionless equations of the third subcase, only one gives rise to an integrable equation of type (\ref{f}), namely, to the third case of Theorem \ref{T1}:
\begin{equation}\label{eq1_subsec43}
u_{xy}=u_x{\triangle_{z\bar z}u}.
\end{equation}
It is shown below that (\ref{eq1_subsec43}) can be derived from equations (\ref{e7}) and (\ref{e8}) under a suitable choice of the parameters. By using the discretization rule explained in Section \ref{sec:3}  we convert equations (\ref{e7}) -- (\ref{e10}) to the following ones:
\begin{multline}  \label{H7}
u_{n,xy} = -\frac{\beta''(u_n)}{\beta'(u_n)} u_{n,x} u_{n,y} + \beta'(u_n) (u_{n+1} - 2 u_n + u_{n-1}) u_{n,y} + \\
+\left( \frac{1}{2} \alpha^2 \beta(u_n) + \frac{3}{2} \alpha \beta'(u_n) \sqrt{(u_{n+1} - u_n)(u_n - u_{n-1})} + \beta''(u_n) (u_{n+1} - u_n) (u_n - u_{n-1}) \right) u_{n,y},	
\end{multline}
\begin{multline}	\label{H8}
u_{n,xy} = -\frac{\beta'(u_n)}{\beta(u_n)} u_{n,x} u_{n,y} +\\
+\bigl( \gamma + \beta(u_n) u_{n,y}  \bigr) \left( u_{n+1} - 2 u_n + u_{n-1} + \frac{\delta}{\beta(u_n)} + \frac{\beta'(u_n)}{\beta(u_n)} (u_{n+1} - u_n) (u_n - u_{n-1}) \right),
\end{multline}
\begin{multline}	\label{H9}
u_{n,xy} = -\frac{\beta''(u_n)}{\beta'(u_n)} u_{n,x} u_{n,y}+ \\
+\gamma e^{\frac{1}{2} \alpha \beta(u_n) + \beta'(u_n) \sqrt{(u_{n+1} - u_n)(u_n - u_{n-1})}} 
\Bigl( \alpha + 2 \beta'(u_n) (u_{n+1} - 2 u_n + u_{n-1}) + \Bigr.\\
\Bigl.+ \alpha \beta'(u_n) \sqrt{(u_{n+1} - u_n)(u_n - u_{n-1})} + 2 \beta''(u_n)(u_{n+1} - u_n)(u_n - u_{n-1}) \Bigr) u_{n,y},
\end{multline}
\begin{multline}	\label{H10}
u_{n,xy} = -\frac{\beta'(u_n)}{\beta(u_n)} u_{n,x} u_{n,y} 
+ \delta e^{\beta(u_n) \sqrt{(u_{n+1} - u_n)(u_n - u_{n-1})}}  \times \\
\times \left( \beta(u_n)(u_{n+1} - 2 u_n + u_{n-1}) + \beta'(u_n) (u_{n+1} - u_n)(u_n - u_{n-1})\right) \left( u_{n,y} + \frac{\gamma}{\beta(u_n)}  \right),
\end{multline}
respectively. These equations are of  the form
\begin{equation}		\label{type2}
u_{n,xy} =   s(u_n) u_{n,x} u_{n,y} + q(u_{n-1}, u_n, u_{n+1}) u_{n,y}
\end{equation}
where $s(u_n) = - \frac{\beta''(u_n)}{\beta'(u_n)}$ for (\ref{H7}), (\ref{H9}), and  $s(u_n) = - \frac{\beta'(u_n)}{\beta(u_n)}$ for (\ref{H8}), (\ref{H10}). In the first case we apply the point transformation $v_n = \beta(u_n) $, that is, $u_n = \varphi(v_n)$, $\varphi = \beta^{-1}$.  In the second case  we use $v_n = \psi(u_n)$ where $\psi' = \beta $, that is, $u_n = \varphi(v_n)$, $\varphi = \psi^{-1}$. Thus, equations (\ref{type2}) take the form
\begin{equation*}  \label{g_uy}
v_{n,xy} = g(v_{n+1}, v_n, v_{n-1}) v_{n,y} + r(v_{n+1}, v_n, v_{n-1})
\end{equation*}
where $\frac{\partial g(v_{n+1}, v_n, v_{n-1})}{\partial v_{n \pm 1}} \neq 0$. Namely,
\begin{multline} \label{HH7}
v_{n,xy} = \beta'^2(u_n) \left( \varphi(v_{n+1}) - 2 \varphi(v_n) + \varphi(v_{n-1}) \right) \varphi'(v_n) v_{n,y} + \\
+ \beta'(u_n) \Bigl( \frac{1}{2} \alpha^2 \beta(u_n) + \frac{3}{2} \alpha \beta'(u_n) \sqrt{(\varphi(v_{n+1}) - \varphi(v_n))(\varphi(v_n) - \varphi(v_{n-1}))} + \Bigr.\\
\Bigl. + \beta''(u_n) (\varphi(v_{n+1}) - \varphi(v_n)) (\varphi(v_n) - \varphi(v_{n-1})) \Bigr) \varphi'(v_n)v_{n,y},	
\end{multline}
\begin{multline}	\label{HH8}
v_{n,xy} = \beta(u_n) \bigl( \gamma + \beta(u_n) \varphi'(v_n) v_{n,y}  \bigr) \Bigl( \varphi(v_{n+1}) - 2 \varphi(v_n) + \varphi(v_{n-1}) + \frac{\delta}{\beta(u_n)} + \Bigr.\\
\Bigl. + \frac{\beta'(u_n)}{\beta(u_n)} (\varphi(v_{n+1}) - \varphi(v_n))(\varphi(v_n) - \varphi(v_{n-1}))\Bigr),
\end{multline}
\begin{multline}	\label{HH9}
v_{n,xy} = \beta'(u_n) \gamma e^{\frac{1}{2} \alpha \beta(u_n) + \beta'(u_n) \sqrt{(\varphi(v_{n+1}) - \varphi(v_n))(\varphi(v_n) - \varphi(v_{n-1}))}} \times \\
\times \Bigl( \alpha + 2 \beta'(u_n) (\varphi(v_{n+1}) - 2 \varphi(v_n) + \varphi(v_{n-1}))  + \alpha \beta'(u_n) \sqrt{(\varphi(v_{n+1}) - \varphi(v_n))(\varphi(v_n) - \varphi(v_{n-1}))} + \Bigr.\\
\Bigl. + 2 \beta''(u_n)(\varphi(v_{n+1}) - \varphi(v_n))(\varphi(v_n) - \varphi(v_{n-1})) \Bigr) \varphi'(v_n)v_{n,y},
\end{multline}
\begin{multline} \label{HH10}
v_{n,xy} = \beta(u_n) \delta e^{\beta(u_n) \sqrt{(\varphi(v_{n+1}) - \varphi(v_n))(\varphi(v_n) - \varphi(v_{n-1}))}}  \times \\
\times \left( \beta(u_n)(\varphi(v_{n+1}) - 2 \varphi(v_n) + \varphi(v_{n-1})) + \beta'(u_n)(\varphi(v_{n+1}) - \varphi(v_n))(\varphi(v_n) - \varphi(v_{n-1}))\right) \times\\
\times \left( \varphi'(v_n) v_{n,y} + \frac{\gamma}{\beta(u_n)}  \right).
\end{multline}
Thus, we can apply Proposition~\ref{P2}. A simple analysis shows that equations (\ref{HH9}), (\ref{HH10}) are never reduced to (\ref{Prop2_eq1}), (\ref{Prop2_eq2}). Let us calculate  mixed partial derivatives of the right-hand sides $\tilde{f}$, $\tilde{\tilde{f}}$ of (\ref{HH7}), (\ref{HH8}), respectively:
\begin{eqnarray}
\tilde{f}_{v_{n+1},v_{n-1}} = -\beta'(u_n) \beta''(u_n)\varphi'(v_{n+1})\varphi'(v_{n-1}) \varphi'(v_n) v_{n,y}, \label{tilde_f} \\
 \tilde{\tilde{f}}_{v_{n+1},v_{n-1}} = - \beta(u_n) \beta'(u_n) \varphi'(v_{n+1})\varphi'(v_{n-1}) \varphi'(v_n) v_{n,y}. \label{cap_f}
\end{eqnarray}
Due to formulae  (\ref{Prop2_eq1}), (\ref{Prop2_eq2}) the equalities $\tilde{f}_{v_{n+1},v_{n-1}} = 0$, $\tilde{\tilde{f}}_{v_{n+1},v_{n-1}} = 0$ should be satisfied. Consequently, we have $\beta'' = 0$ in case (\ref{tilde_f}) and $\beta' = 0$ in case (\ref{cap_f}). Thus, equations (\ref{H7}), (\ref{H8}) take the form
\begin{eqnarray}
&u_{n,xy} = c_1 (u_{n+1} - 2u_n + u_{n-1})u_{n,y},  \label{sec43_last}\\
&u_{n,xy} = (\gamma + c_1 u_{n,y})(u_{n+1} - 2 u_n + u_{n-1} + \frac{\delta}{c_1 u + c_2}). \label{sec43_last2}
\end{eqnarray}
Due to Proposition~\ref{P2} both of these equations have to concide with equation (\ref{Prop2_eq2}). Thus, we conclude that $\delta = 0$ in (\ref{sec43_last2}). Then (\ref{sec43_last2}) is reduced to (\ref{Prop2_eq2}) by the changes of  variables $u_n \rightarrow u_n - \frac{\gamma}{c_1} y$, $x \rightarrow \frac{x}{c_1}$.
Equation (\ref{sec43_last}) is reduced to (\ref{Prop2_eq2}) by the change of  variables $x \rightarrow \frac{x}{c_1}$.
Ultimately, both of them are equivalent to the third case of Theorem \ref{T1}.

\subsection{Subcase 4}
\label{sec:s4}

It is readily seen that equations of Subcase 4 give rise to lattice equations of type (\ref{Prop1_eq_gen}). For this class  the complete classification is given in \cite{HP1} (see Proposition~\ref{P1} above). Hence investigation of the equations  (\ref{e11})--(\ref{e16}) is straightforward. This will result in the last three cases of Theorem~\ref{T1}.
Let us first concentrate on the equations (\ref{e11})--(\ref{e13}), which are integrable for a certain choice of parameters. 
They correspond to lattice  equations of the form:
\begin{equation}
u_{n,xy} = \frac{2(u_{n+1} - 2 u_n + u_{n-1}) + (4 \beta'(u_n) - \alpha) \sqrt{(u_{n+1} - u_n)(u_n - u_{n-1})} + 2 \beta \beta' - \alpha \beta}{2 \left( \sqrt{(u_{n+1} - u_n)(u_n - u_{n-1})} + \beta(u_n) \right)^2} u_{n,x} u_{n,y},  \label{H11}
\end{equation} 
\begin{multline}	\label{H12}
u_{n,xy} = \frac{u_{n,x} u_{n,y} + \beta(u_n) u_{n,x}}{\left( \sqrt{(u_{n+1}-u_n)(u_n - u_{n-1})} + \gamma \beta(u_n) \right)^2} (u_{n+1} - 2 u_n + u_{n-1}) + \\
+ \frac{(4 \gamma \beta'(u_n) - \alpha)\sqrt{(u_{n+1} - u_n)(u_n - u_{n-1})} + 2 \gamma^2 \beta(u_n) \beta'(u_n) - \alpha \gamma \beta(u_n) }{2 \left( \sqrt{(u_{n+1} - u_n)(u_n - u_{n-1})} + \gamma \beta(u_n) \right)^2} u_{n,x} u_{n,y} - \\
- \frac{2 \beta'(u_n)(u_{n+1} - u_n)(u_n - u_{n-1}) + \alpha \beta(u_n) \sqrt{(u_{n+1} - u_n)(u_n - u_{n-1})} + \alpha \gamma \beta^2(u_n) }{ 2   \left( \sqrt{(u_{n+1} - u_n)(u_n - u_{n-1})} + \gamma \beta(u_n)   \right)^2} u_{n,x},
\end{multline}
\begin{multline}		
u_{n,xy} = \frac{(u_x + \beta(u_n))(u_y + \delta \beta(u_n))}{\left( \sqrt{(u_{n+1} - u_n)(u_n - u_{n-1})} + \gamma \beta(u_n) \right)^2} (u_{n+1} - 2 u_n + u_{n-1}) + \\
+ \frac{(4 \gamma \beta'(u_n) - \alpha)\sqrt{(u_{n+1} - u_n)(u_n - u_{n-1})} + 2 \gamma^2 \beta(u_n) \beta'(u_n) - \alpha \gamma \beta(u_n)}{2 \left( \sqrt{(u_{n+1} - u_n)(u_n - u_{n-1})} + \gamma \beta(u_n) \right)^2} u_{n,x} u_{n,y} - \\
- \frac{2 \beta'(u_n) (u_{n+1} - u_n)(u_n - u_{n-1}) + \alpha \beta(u_n) \sqrt{(u_{n+1} - u_n)(u_n - u_{n-1})} + \alpha \gamma \beta^2(u_n)}{2 \left( \sqrt{(u_{n+1} - u_n)(u_n - u_{n-1})} + \gamma \beta(u_n) \right)^2} \times \\
\times \left( u_y + \delta u_x + \delta \beta(u_n) \right).		\label{H13}
\end{multline}
More precisely, equations (\ref{H11}), (\ref{H12}) are integrable only when $\alpha = \beta = 0$, and then they coincide with (\ref{Prop1_eq1}). Equation (\ref{H13}) is more generic, its integrable cases are obtained by setting:
\begin{eqnarray*}
&&i)\, \alpha=0,\, \beta = 0;\\
&&ii)\, \alpha = 0,\, \beta = 0, \, \gamma = 0, \, \delta = 0;\\
&&iii)\, \alpha = 0, \, \gamma = 0, \, \delta \neq 0, \, \beta = -u_n;\\
&&iv)\,  \alpha = 0, \, \gamma = 0, \, \delta \neq 0, \, \beta = -u^2_n - 1.
\end{eqnarray*}
In the first two cases we get equation (\ref{Prop1_eq1}); for $iii)$ and $iv)$ equation (\ref{H13}) is reduced to (\ref{Prop1_eq2}) and, respectively, to (\ref{Prop1_eq3}). 
As for equations (\ref{e14})--(\ref{e16}), they are not integrable since they don't satisfy the integrability conditions derived in \cite{HP1}.

Modulo elementary equivalence transformations, the final list of lattice equations (\ref{f}) passing both tests is as follows:
\begin{itemize}
\item[$1)$] $u_{n,xy} = e^{u_{n+1} - 2 u_n + u_{n-1} },$
\item[$2)$] $u_{n,xy} = e^{u_{n+1}} - 2 e^{u_n} + e^{u_{n-1}},$
\item[$3)$] $u_{n,xy} = \left(u_{n+1} - 2 u_n + u_{n-1}  \right) u_{n,x}, $
\item[$4)$] $u_{n,xy}=\alpha_nu_{n,x}u_{n,y}, \quad \alpha_n = \frac{1}{u_n - u_{n-1}} - \frac{1}{u_{n+1}-u_n}=\frac{u_{n+1} - 2 u_n + u_{n-1}}{(u_{n+1}-u_n)(u_n - u_{n-1})},
$
\item[$5)$] $u_{n,xy} = \alpha_n(u_{n,x} - u_n)(u_{n,y}-u_n) +u_{n,x}+ u_{n,y} -u_n, $
\item[$6)$] $u_{n,xy} = \alpha_n(u_{n,x} - u^2_n - 1)(u_{n,y} - u^2_n - 1) + 2 u_n(u_{n,x}+u_{n,y}-u^2_n - 1).$
\end{itemize}
All of them are equivalent to the normal forms of Theorem \ref{T1}. 

\medskip

\noindent{\bf Remark.} Note that equation 2, which is not exactly of the form (\ref{f}), can be transformed into this form: rewriting it as $u_{xy}=\triangle_{z\bar z}e^u$ and setting $u=\ln v$ we obtain
$$
v_{xy}=\frac{v_xv_y}{v}+v\triangle_{z\bar z}v.
$$
On the other hand, the integrable equation $u_{n,xy} = e^{u_{n+1} - u_n} - e^{u_n - u_{n-1}}$, which is in fact equivalent to equation 2, is formally not of type (\ref{f}) and therefore does not appear on the list. Integrable lattice (\ref{Prop2_eq1}) is also absent from the list for the same reason.

\section*{Acknowledgements}
We thank Maxim Pavlov for useful discussions. The research of EVF was  supported by the EPSRC grant  EP/N031369/1.


\end{document}